\documentclass[letter, 10 pt, conference]{ieeeconf}

\IEEEoverridecommandlockouts                              

\usepackage{cite}
\usepackage{amsmath,amssymb,amsfonts}
\usepackage{algorithmic}
\usepackage{graphicx}
\usepackage{textcomp}
\usepackage{xcolor}
\usepackage{graphicx}
\usepackage{xspace}
\usepackage{listings}
\usepackage{amsfonts}
\usepackage{tabularx}
\usepackage{caption}
\usepackage{subcaption}
\usepackage{algorithm}
\usepackage{algorithmic}
\usepackage{cite}
\usepackage{hyperref}
\usepackage[capitalize]{cleveref}
\usepackage{mathrsfs}
\usepackage{bbm}
\usepackage{bm}
\usepackage{tikzsymbols}

\usepackage{amsthm}
\usepackage{booktabs}
\usepackage{mathtools}

\theoremstyle{definition}
\newtheorem{definition}{Definition}[section]
\newtheorem{prop}{Proposition}[section]

\newtheorem{example}{Example}[section]
\newtheorem{theorem}{Theorem}[section]

\title{\LARGE \bf
Improving the Feasibility of Moment-Based \\ Safety Analysis for Stochastic Dynamics
}

\author{Peter Du, Katherine Driggs-Campbell, and Roy Dong
\thanks{P. Du, K. Driggs-Campbell, and R. Dong are with the Department of Electrical and Computer Engineering at the University of Illinois at Urbana-Champaign. email: \{peterdu2, krdc, roydong\}@illinois.edu}
\thanks{This work has been submitted to the IEEE for possible publication. Copyright may be transferred without notice, after which this version may no longer be accessible.}%
}

\begin{document}

\maketitle

\begin{abstract}

Given a stochastic dynamical system modelled via stochastic differential equations (SDEs), we evaluate the safety of the system through characterizations of its exit time moments. We lift the (possibly nonlinear) dynamics into the space of the occupation and exit measures to obtain a set of linear evolution equations which depend on the infinitesimal generator of the SDE. Coupled with appropriate semidefinite positive matrix constraints, this yields a moment-based approach for the computation of exit time moments of SDEs with polynomial drift and diffusion dynamics. However, the existing moment approach suffers from drawbacks which impede its applicability to the analysis of higher dimensional physical systems. To apply the existing approach, the dynamics of the systems are limited to polynomials of the state --- excluding a large majority of real world examples. Computational scalability is also poor as the dimensionality of the state increases, largely due to the combinatorial growth of the optimization program. In this paper, we propose changes to the existing moment method to make feasible the safety analysis of higher dimensional physical systems. The restriction to polynomial dynamics is lifted by using a state augmentation method which allows one to generate the evolution equations for a broader class of nonlinear stochastic systems. We then reformulate the constraints of the optimization to mitigate the computational limitations associated with an increase in state dimensionality. We employ our methodology on two example processes to characterize their safety via expected exit times and demonstrate the ability to handle multi-dimensional physical systems that were previously unsupported by the existing SDP method of moments. 

\end{abstract}

\section{Introduction}
\label{sec:introduction}

Safety verification is an important step in ensuring dynamical systems perform in ways that designers intend while also mitigating the risks associated with unplanned behavior. In the deterministic scenario, reachability can be applied to the safety verification problem and produce a boolean result that predicts whether a system enters a set of unsafe configurations at some point in the future. 

In the stochastic setting, we seek an analog to this boolean safety statement through the characterization of the distribution of exit times for a system governed by stochastic dynamics. We consider a state space partitioned into two sets $S$ (safe), and $S^c$ (unsafe), and denote the exit time $\tau$ as the first time that the system, starting from $x_0 \in S$, reaches a state $x_{\tau} \in S^c$.

Our work applies a moment-based approach \cite{lasserre2008sdpGPM} to obtain exit time moments for stochastic dynamical systems modelled via stochastic differential equations (SDEs). We formulate the computation of the exit time as an infinite-dimensional convex optimization over the space of measures through the use of a linear evolution equation on the moments of the measures. Following the approach in \cite{lasserre2004sdpvslp}, we use semidefinite programming (SDP) to compute bounds on the expected duration of safe operation. This moment-based method has previously seen successful application in financial instrument pricing \cite{lasserre2006PricingAC} and optimal control \cite{henrion2008nonlinear}. In this paper, we seek to use the moment method to study the safety properties of dynamical systems through the distributions of their exit time from a safe region. 

In the context of probabilistic systems, it is often intuitive to characterize the safety of a system through some probability $p$ of violating some set of constraints. This can further be combined with an appropriate cost function to generate a risk metric which can aid in the determination of a suitable $p$. Prior works such as \cite{esfahani2016reach}, \cite{wisniewski2015safety}, and \cite{wisniewski2020pSafe} have looked at obtaining a set of initial states that satisfy a safety requirement based on $p$. In our work, we consider the problem from an alternative perspective in which we characterize the safety of a \textit{given} set of initial states through their exit time distribution. 

We consider a semidefinite programming based approach for the numerical computation of exit time moments for Markov processes \cite{lasserre2004sdpvslp}. More precisely, we focus our attention to dynamical systems modelled by stochastic differential equations. Applied to SDEs, the existing method is largely restricted to systems with polynomial drift and diffusion dynamics as it relies on considering processes whose infinitesimal generators map the monomials into polynomials. In particular, the inability to support sinusoidal dynamics prevents application to a large number of physical systems. Indeed, sinusoids can be found in the dynamics of virtually all robotic systems where the forces on the system are applied at varying angles with respect to a chosen coordinate frame. Examples include robotic arms, tracked robots, wheeled robots, and quadcopters. The existing SDP method of moments also suffers from poor computational scalability as the dimensionality of the system and complexity of the safe set grows. The combinatorial growth of the number of moments results in matrix constraints that quickly drive computational costs upward. 

In this paper, we propose the use of an SDP method of moments to analyze the safety characteristics of physical systems via their distribution of exit times from a safe set \cite{lasserre2004sdpvslp}. The following contributions are presented: 

\begin{enumerate}
  \item We discuss the limitations of the existing SDP moment method when applied to high dimensional systems and the restrictions that exclude a large class of applications.
  \item To analyze the safety of complex dynamics, we introduce a state augmentation method and describe the trade-offs in dimensionality. We then propose a reformulation of the matrix constraints used to develop the semidefinite program and replace large sequences of positive semidefinite (PSD) matrix constraints with smaller sets of scalar equality constraints.
  \item We provide examples of our approach in handling a broader class of real world dynamics and demonstrate the ability to compute accurate bounds on the duration of safe behavior by considering higher degree moments.  
\end{enumerate}

This paper is organized as follows: \cref{sec:background} reviews related works and techniques in stochastic system safety verification and the exit time problem. \cref{sec:system_model} describes the SDE system model and notation. \cref{sec:computation} presents the existing SDP method for exit time moment computation. \cref{sec:augmentation} introduces the state augmentation method and \cref{sec:reduced_constraint} presents a reformulation of the optimization constraints. In \cref{sec:experiments} we present numerical examples to characterize system safety. Lastly, we conclude in \cref{sec:conclusion}.

\section{Related Work}
\label{sec:background}
 
Fisac et al. \cite{fisac2019generalsafety} proposes a general safety framework for controlled dynamical systems subject to deterministic (but unknown), state-dependent disturbances $d(x)$. Hamilton-Jacobi reachability methods are combined with Bayesian inference to generate a safe control policy that ensures the system remains within a predefined safe set. The state-dependent disturbances are assumed to be drawn from a Gaussian process and new observations are incorporated as the system samples additional disturbances. Equipped with a model of $d(x)$, the authors construct a probabilistic bound over the space of disturbances and incorporate this into the computation of a non-static safe set and optimal safety controller. 

In the stochastic dynamics setting, several studies have been proposed in which an initial set satisfying certain (static) safety conditions is obtained. Reachability of stochastic systems has been studied through stochastic viability and target problems \cite{aubin2000stochasticviability, soner2002stochastictarget}. A connection between stochastic optimal control (the exit-time problem) and the reach-avoid problem for controlled diffusion processes is presented by Esfahani et al. \cite{esfahani2016reach}. Here, the work proposes a method for computing the set of initial states where there exists an admissible control scheme such that the system hits a desired set prior to entering an avoid set. Notably, the set of initial states are characterized by the super level sets of the viscosity solution of a suitable Hamilton-Jacobi-Bellman equation. 

Occupation measure approaches for computing the region of attraction (ROA) in deterministic systems have been studied in \cite{korda2012Inner} and \cite{korda2014ROA}. Korda et al. \cite{korda2012Inner} analyzes the ROA for deterministic nonlinear (polynomial) dynamical systems and show a linear programming approach for approximating the region. The authors propose an optimization over occupation measures and describe the nonlinear system dynamics through an equivalent linear evolution equation over measures. Using a similar analysis over measures, the notion of $p$-safety for stochastic systems is developed by Wisniewski et al. \cite{wisniewski2015safety, wisniewski2020pSafe} and is most closely related to our methodology. Under the notion of $p$-safety, one analyzes the set of initial conditions for which the system is safe with probability at least $p$. A starting state $s_0$ is said to be $p$-safe if the probability of trajectories initiating from $s_0$ reaching an unsafe set is less than $1-p$.  The evolution equation of the occupation measure is applied to stochastic polynomial dynamics. Similar to the approach we use, $p$-safety employs the linear evolution to link the initial, final, and occupation measures of the stochastic system and formulates an infinite-dimensional optimization that is solved by the generalized moment method (returning the largest $p$-safe set) \cite{wisniewski2020pSafe}.

Numerical approaches for the exit time problem applied to Markov processes have also been studied \cite{lasserre2004sdpvslp} \cite{helmes2001lpmoments}. At their core, the methods aim to characterize the exit time problem via an infinite-dimensional convex program subject to constraints derived from the martingale characterization of Markov processes \cite{ethier2009markovprocess}. Relating to the characterization found in \cite{korda2012Inner}, the evolution of functionals over these processes is formulated using the moments of the occupation and exit location measures (described through the \textit{basic adjoint equation}), reducing the analysis to the space of moment sequences of these measures. An SDP approach to the moment problem is proposed by Lasserre et al. \cite{lasserre2008sdpGPM}. By replacing the LP Hausdorff moment constraints with SDP constraints, the approach showed increased computational efficiency and accuracy. Over the following decade, the SDP approach was utilized in a number of applications in a diverse set of fields, ranging from chemical/biomolecuar dynamics to finance and economics \cite{lasserre2006PricingAC} \cite{jubril2013econIEEE} \cite{yukata2020CMEkinetics} \cite{garrett2018CMEdynamicbounds}. In order to derive the appropriate linear constraints, these works consider polynomial dynamics which are often naturally afforded due to the underlying system characteristics. 

Scalability issues are well known and discussed in the literature applying the SDP approach to solve real world systems. For instance, the moment approach is used to find bounds of the survival time of chemical systems using PSD constraints \cite{matsunaga2021CME}. In the examples presented, a low moment degree is used and the authors make explicit the limitations on scalability due to the combinatorial growth of moments when the state dimension increases. Numerical stability is also questioned as higher moments quickly begin to differ by multiple orders of magnitude. Likewise, researchers have applied the SDP method to the analysis of economic-emission dispatch and discuss the limitations to small/medium sized problems due to the quantity of computations involving a large number of high dimensional matrices \cite{jubril2014economicSDP}.

In this paper, we consider the SDP based moment method for evaluating the safety of stochastic dynamical systems and use state augmentation to extend beyond polynomial dynamics. By reformulating the constraints of the SDP, we aim to reduce the number of large dimensional matrices present in the optimization and allow for applications with a greater number of moments that yield bounds with higher accuracy. 
\section{System Model}
\label{sec:system_model}
\subsection{Notation}
\label{sec:notation}

For two values $a,b \in \mathbb{R}$, we define $a \ \wedge \ b \vcentcolon= \min\{a,b\}$. Given a set $A$, we denote its complement by $A^c$ and its boundary by $\partial A$. The Borel $\sigma$-algebra on a topological space $\mathcal{A}$ is denoted by $\mathscr{B}(\mathcal{A})$ and for $B \in \mathscr{B}(\mathcal{A})$, the indicator function is denoted by $\mathbbm{1}_B$ and defined as $\mathbbm{1}_B(x)=1$ if $x \in B$ and 0 otherwise. The support of a measure $\mu$ on a measurable space $(\mathcal{A}, \mathscr{B}(\mathcal{A}))$ is denoted by $\text{supp}(\mu)$. For a process $X = (X_t)_{t \geq 0}$ described via stochastic differential equation in $\mathbb{R}^n$, we denote $\hat{X}$ as the state augmented version of $X$ given by an SDE in $\mathbb{R}^{n+s}, s>0$. The set of integers $\{1,2, \dots N\}$ is denoted by $[N]$. We represent the $n$-dimensional multi-index $\alpha$ as a tuple such that $\alpha=(\alpha_1, \alpha_2, \ldots \alpha_n)$. The set of such $n$-dimensional multi-indexes is denoted by $\mathbb{N}^n$. Lastly, the monomial with degree corresponding to the multi-index $\alpha$ is given by $(x_1, x_2,...,x_n)^\alpha$ such that $(x_1, x_2,...,x_n)^\alpha$ = $x_1^{\alpha_1} x_2^{\alpha_2} ... x_n^{\alpha_n}$.

\subsection{Stochastic Differential Equations}
\label{sec:model}

We consider the $\mathbb{R}^n$ valued stochastic differential equation:

\begin{equation}
dX_t = h(X_t, t)dt + \sigma (X_t, t)dB_t\\[7pt]
\label{eq:sde}
\end{equation}

\noindent where $X_t \in E \subseteq \mathbb{R}^n$, $X_0 = x_0$ is known, $0 \leq t \leq T$, and $T>0$. Let $B_t$ be a standard $d$-dimensional Brownian motion and $dB_t$ represent its differential form. Let the functions $h : \mathbb{R}^n \times [0,T] \rightarrow \mathbb{R}^n$ and $\sigma : \mathbb{R}^n \times [0,T] \rightarrow \mathbb{R}^{n \times d}$ represent the drift and diffusion terms of the SDE, respectively. Furthermore, let the functions be measurable and satisfy the space variable growth condition:

\vspace{-0.3cm}
$$
|h(x,t)| + |\sigma(x,t)| \leq C(1+|x|) \quad x \in \mathbb{R}^n, t \in [0,T]\\[3pt]
$$

\noindent for some constant $C$, as well as the space variable Lipschitz condition:

\vspace{-0.5cm}
$$
|h(x,t) - h(y,t)| + |\sigma(x,t) - \sigma(y,t)| \leq D|x-y|
$$
$$
x,y\in\mathbb{R}^n, t\in[0,T] \\[6pt]
$$
\vspace{-0.7cm}

Under these circumstances, the stochastic differential equation (\ref{eq:sde}) has a unique time continuous solution starting at time $t$ and state $x_0$ \cite[Theorem~5.2.1]{oksendal2003sde}. In addition, the stochastic process $X = (X_t)_{t\geq 0}$, given by the SDE (\ref{eq:sde}), with initial condition $X_0 = x_0$ with probability one, is a Markov process with continuous sample paths \cite[Theorem~7.1.2]{oksendal2003sde}.

We consider a state space $E \subseteq \mathbb{R}^n$ that is partitioned into two sets: $S$ and $S^c$. Here, $S$ is an open and bounded safe set and $S^c = E - S$ is its complement (unsafe set). In this paper, $\tau$ is a stopping time defined with respect to $S^c$ and is the minimum of the first time that the process $X$ reaches the unsafe set: 

\begin{equation}
    \tau = \inf\{t \mid X_t \in S^c\} \\[7pt]
    \label{eq:tau}
\end{equation}

Throughout the rest of this paper, we will be concerned with a finite exit time $\tau \wedge T$. Intuitively speaking, if the exit time $\tau \wedge T$ is strictly less than $T$, then the system has become unsafe within the time horizon we are concerned with. While if $\tau \wedge T = T$, the system has stayed safe almost surely for the entire finite duration we are examining. 
\section{Computation of Exit Time Moments}
\label{sec:computation}


In this section, we describe methods for exit moment computation through an infinite-dimensional optimization program \cite{lasserre2004sdpvslp} \cite{helmes2001lpmoments}.  

\subsection{Linear Evolution Equation}
\label{sec:mt_constraints}
Let $(X_t)_{t \geq 0}$ be a time-homogeneous diffusion in $\mathbb{R}^n$ such that its dynamics are given by the following SDE:

\begin{equation}
    dX_t = h(X_t)dt + \sigma(X_t)dB_t
    \label{eq:th_sde}
\end{equation}

We use $P^x$ to denote the probability laws of $(X_t)_{t \geq 0}$ such that $P^x$ gives the distribution of $(X_t)_{t \geq 0}$ when $X_0=x$. Furthermore, let $\mathbb{E}^x$ denote the expectation w.r.t the probability law $P^x$. The infinitesimal generator $A$ of $X_t$ is defined as \cite[Definition~7.3.1]{oksendal2003sde}:

\begin{equation*}
    Af(x) = \lim_{t\downarrow 0} \frac{\mathbb{E}^x[f(X_t)]-f(x)}{t}\\[6pt]
\end{equation*}

The set of functions $f:\mathbb{R}^n \rightarrow \mathbb{R}$ such that the above limit exists for all $x \in \mathbb{R}^n$ is the denoted as the domain $D(A)$.

The generator of a time-homogeneous It\^{o} diffusion in $\mathbb{R}^n$ for twice differentiable continuous $f$ is \cite[Theorem~7.3.3]{oksendal2003sde}:

\begin{equation}
Af(x) = \sum_i h_i(x)\frac{\partial f}{\partial x_i} + \frac{1}{2}\sum_{i,j}(\sigma \sigma^\intercal)_{i,j}
(x) \frac{\partial ^2 f}{\partial x_i \partial x_j}
\label{eq:gen_th_ito}
\end{equation}
\label{prop:ito_diff_generator}

The dynamics of (\ref{eq:th_sde}) are lifted to the space of measures to define a set of linear evolution equations. The process $(X_t)$ given by the SDE (\ref{eq:th_sde}) satisfies the martingale problem where: 

\begin{equation}
f(X_t) - f(X_0) - \int_0^t Af(X_s)ds\\[7pt]
\label{eq:martingale_problem}
\end{equation}

\noindent is a martingale for all test functions $f \in D(A)$. The first moment of the exit time $\tau \wedge T$ remains finite. Combined with the martingale property of (\ref{eq:martingale_problem}), we have: 

\begin{equation}
\mathbb{E}^{x_0}[f(X_{\tau \wedge T})] - \mathbb{E}^{x_0}[f(X_0)] - \mathbb{E}^{x_0} \left[ \int^{\tau \wedge T}_0 Af(X_s)ds \right] = 0\\[7pt]
\label{eq:mt_problem_expectation}
\end{equation}
Here, the notation $E^{x_0}$ serves to emphasize that $X_0 = x_0$.

Let $\mu_0$ be the expected occupation measure up to the exit time $\tau \wedge T$ of the process $(X_t)$, and $\mu_1$ be its exit location distribution:

\begin{equation*}
    \begin{gathered}
        \mu_0(B) = \mathbb{E} \int_0^{\tau \wedge T} \mathbbm{1}_B(X_t)dt \\[6pt]
        \mu_1(B) = \mathbb{P}(X_{\tau \wedge T} \in B)\\[7pt]
    \end{gathered}
\end{equation*}

The measures $\mu_0$ and $\mu_1$ are supported on the safe set $S$ and safe set boundary $\partial S$, respectively ($\text{supp}(\mu_0)=S$, $\text{supp}(\mu_1)=\partial S$). Equation (\ref{eq:mt_problem_expectation}) is now rewritten as:

\begin{equation}
\int_{\partial S}f(x)\mu_1(dx) - f(x_0) - \int_{S}Af(x)\mu_0(dx) = 0\\[5pt]
\label{eq:basic_adjoint_equation}
\end{equation}

\noindent for every test function $f \in D(A)$ and $X_0 = x_0 \in S$.
Equation (\ref{eq:basic_adjoint_equation}) represents a linear evolution equation linking the occupation and exit measures of the process $(X_t)$, also referred to as the \textit{basic adjoint equation} \cite{helmes2001lpmoments}. 

The moments of the measures $\mu_0$ and $\mu_1$ are given by:

$$
m_i = \int_S x^i \mu_0(dx) \quad \text{and} \quad b_i = \int_{\partial S} x^i \mu_1(dx)\\[7pt]
$$

\noindent where each $i \in \mathbb{N}^n$ is an $n$-dimensional multi-index and $x^i = x_1^{i_1} x_2^{i_2} \cdot\cdot\cdot x_n^{i_n}$. Notice the first moment of the exit time is $m_0$. Under processes where monomial test functions $f$ produce a polynomial infinitesimal generator $Af$, the conditions imposed by the basic adjoint equation are further relaxed from the space of all functions $f \in D(A)$ to all monomials $f$, and expressed through the sequence of moments of $\mu_0$ and $\mu_1$: $[m_i]_{i\in\mathbb{N}^n}$ and $[b_i]_{i\in\mathbb{N}^n}$. The relaxed condition:

\begin{equation}
\sum_j c_j(i) \cdot m_j + x_0^i - b_i = 0 
\label{eq:mt_const}
\end{equation}

\noindent is imposed for every $i \in \mathbb{N}^n$ and monomial $f(x) = x^i$. The condition (\ref{eq:mt_const}) gives a set of linear constraints involving the moments of the exit time and exit distribution.

\subsection{SDP Moment Constraints}
\label{sec:sdp_moment_constraints}
The martingale constraints (\ref{eq:mt_const}) alone are not able to guarantee the sequences $[m_i]$ and $[b_i]$ are \textit{moment} sequences with respect to the appropriate occupation and exit measures. In order to enforce that the sequences are \textit{moment} sequences, additional conditions must be imposed. Helmes et al. \cite{helmes2001lpmoments} considers linear moment constraints while Lasserre et al. \cite{lasserre2004sdpvslp} derives SDP conditions. We choose to use the SDP conditions as they have been shown to provide greater precision and reduced computational requirements. 

Let $[m_\alpha]$ be a sequence where $\alpha\in\mathbb{N}^n$ is a multi-index. The sequence is sorted according to the graded lexicographic order where $\alpha$ represents a monomial $x^\alpha$.

\begin{definition}[Graded Lexicographic Order]
An $n$-dimensional multi-index $\gamma$ is represented by a tuple:
$\gamma=(\gamma_1, \gamma_2, \dots \gamma_n) \in \mathbb{N}^n$. 
The degree of $\gamma$ is given by:

\begin{equation*}
    \text{Deg}(\gamma)=\sum_{i=1}^n\gamma_i \\[6pt]
\end{equation*}

\noindent A multi-index $\alpha$ precedes another multi-index $\beta$ in graded lexicographic order if $\text{Deg}(\alpha) < \text{Deg}(\beta)$ and $\alpha, \beta \in \mathbb{N}^n$. If $\text{Deg}(\alpha) = \text{Deg}(\beta)$, $\alpha$ precedes $\beta$ if the leftmost non-zero entry of the element wise difference $\alpha - \beta$ is positive. 
\end{definition}

\begin{example}
The 3-dimensional indexed moment sequence $[m_\alpha]$, $\alpha \in \mathbb{N}^3$ is given by:

\begin{equation*}
\begin{gathered}[]
    [m_\alpha] = [m_{000}, m_{100},  m_{010},  m_{001},  m_{200},\\
    m_{110},  m_{101},  m_{020},  m_{011},  m_{002}, ...]     
\end{gathered}
\end{equation*}
\end{example}

Given a moment sequence, the moment matrix $M_k(m)$ is defined as follows:

\begin{equation*}
    M_k(m)(i,j) = m_{\alpha + \beta}    
\end{equation*}

\noindent where
\begin{equation*}
\begin{gathered}[]
    M_k(m)(1,j) = m_{\alpha} \\
    M_k(m)(i,1) = m_{\beta} \\[7pt]
\end{gathered}
\end{equation*}

In other words, the top most row and left most column of $M_k(m)$ (i.e. $M_k(m)(0,\cdot)$ and $M_k(m)(\cdot, 0)$) consist of the elements of $[m_\alpha]$ up to degree $k$. 

\begin{example}
Let $x \in \mathbb{R}^2$. The second degree moment matrix $M_2(m)$ is given by:
\vspace{0.25cm}
\begin{equation*}
M_2(m) = 
\begin{bmatrix}
m_{00} & m_{10} & m_{01} & m_{20} & m_{11} & m_{02} \\
m_{10} & m_{20} & m_{11} & m_{30} & m_{21} & m_{12} \\
m_{01} & m_{11} & m_{02} & m_{21} & m_{12} & m_{03} \\
m_{20} & m_{30} & m_{21} & m_{40} & m_{31} & m_{22} \\
m_{11} & m_{21} & m_{12} & m_{31} & m_{22} & m_{13} \\
m_{02} & m_{12} & m_{03} & m_{22} & m_{03} & m_{04} \\
\end{bmatrix}    
\end{equation*}
\end{example}

\vspace{0.35cm}

Next the localizing matrix $M_k(qm)$ is defined with respect to a polynomial $q$. Let $\beta(i,j)$ be the multi-index of the $i,j$th entry of the moment matrix $M_k(m)$ and let $[q_\alpha]$ be the vector of coefficients of the polynomial $q$ in graded lexicographic order. The entries of the localizing matrix is then given by:

\begin{equation*}
M_k(qm)(i,j) = \sum_\alpha q_\alpha \cdot m_{\beta(i,j) + \alpha}
\end{equation*}

\begin{example}
Let $x \in \mathbb{R}$ and $q(x) := 1+x^2+x^4$. The first degree localizing matrix $M_1(qm)$ is given by:
\vspace{0.25cm}
\begin{equation*}
M_1(qm) = 
\begin{bmatrix}
m_0+m_2+m_4 & m_1+m_3+m_5 \\
m_1+m_3+m_5 & m_2+m_4+m_6 \\
\end{bmatrix}\\[10pt]
\end{equation*}
\end{example}

\subsection{Optimization Program}
\label{sec:optimization_program}
The upper and lower bounds of the expected exit time $\mathbb{E}[\tau]$ of the system (\ref{eq:th_sde}) is computed through the following semidefinite program \cite{lasserre2004sdpvslp}:

\vspace{0.5cm}
\noindent \textbf{Optimization I (Original Constraints)} 
\begin{alignat*}{3}
& \text{Maximize}   & &(\textit{resp. Minimize}):& m_0\\
\cline{1-6}
&  \text{Subject to:}&   &\hspace{0.7cm}x_0^k + \sum_{i \in \mathbb{N}^{n}}c_i(k) \cdot m_i - b_k= 0&\\
&   &   &\hspace{0.7cm}M_k(m) \succcurlyeq 0& \\
&   &   &\hspace{0.7cm}M_k(b) \succcurlyeq 0&\\
&   &   &\hspace{0.7cm}M_k(q_0m) \succcurlyeq 0&\\
&   &   &\hspace{0.7cm}M_k(q_1b) \succcurlyeq 0&\hspace{0.1cm}\forall k \leq K\\
\end{alignat*}

\noindent where $M_k(m)$ and $M_k(q_0m)$ are the moment and localizing matrices corresponding to the moment sequence of $\mu_0$, and $M_k(b)$ and $M_k(q_1b)$ are the moment and localizing matrices corresponding to the moment sequence of $\mu_1$. The polynomials $q_0, q_1$ with which the localizing matrices are defined with respect to are derived from the semi-algebraic sets $E_1 := \{x \in \mathbb{R}^d \mid q_0(x) \geq 0\}, E_2 := \{x \in \mathbb{R}^d \mid q_1(x) \geq 0\}$, such that the measures $\mu_0, \mu_1$ are supported on $E_1, E_2$, respectively. To make the program numerically tractable, the optimization is restricted to a finite number of moments $K$.

\section{State Space Augmentation}
\label{sec:augmentation}

In this section we present our state augmentation method using redundant states to support non-polynomial system dynamics. We then discuss the trade-offs in computational complexity when considering higher dimensional state spaces. 

\subsection{Time-Dependent It\^{o} Diffusion}
\label{sec:space_time_system}
In order to compute higher order moments and ensure a finite exit time, the time dimension must be included within the state. The new state $\hat{X}_t \in \mathbb{R}^{n+1}$ is given as $\hat{X}_t = [X_t, t]^{\intercal}$, with dynamics:

\begin{equation}
\begin{aligned}
d\hat{X}_t &= [h(X_t), 1]^\intercal dt + [\sigma(X_t), 0]^\intercal dB_t \\[3pt]
&= \hat{h}(\hat{X}_t) dt + \hat{\sigma}(\hat{X}_t)dB_t\\[6pt]
\end{aligned}
\label{eq:t_aug_sys}
\end{equation}

$\hat{X} = (\hat{X}_t)_{t \geq 0}$ is now an It\^{o} diffusion in $\mathbb{R}^{n+1}$ with initial condition $\hat{x}_0 = (x_0, 0)$. Recall that \cref{eq:mt_problem_expectation} requires consideration of a finite exit time. To address system dynamics which may stay within the safe set for all time $t \in [0,\infty)$, we use a finite time horizon $T$. The safe set of the SDE (\ref{eq:t_aug_sys}) is then $\hat{S} = S \times [0,T]$, which guarantees a finite exit time. Applying the evolution equation (\ref{eq:basic_adjoint_equation}) to $\hat{X}$ yields the basic adjoint equation:

\vspace{0.1cm}
\begin{equation*}
\begin{gathered}
    \int_{\partial \hat{S}}f(x,s)\mu_1(dx \times ds) - f(x_0, 0) \\
     -\int_{\hat{S}}Af(x,s)\mu_0(dx \times ds) = 0 \\[5pt]
\end{gathered}
\end{equation*}

In view of the SDE (\ref{eq:t_aug_sys}) and generator (\ref{eq:gen_th_ito}) associated with the It\^{o} diffusion, we observe that the operator $A$ is composed of differential and summation operations (with respect to the state variables). Thus stochastic dynamics with both polynomial drift $\hat{h}(\cdot)$ and diffusion $\hat{\sigma}(\cdot)$ satisfy the above assumption. As in Section \ref{sec:mt_constraints}, the basic adjoint equation conditions are relaxed from all $f \in D(A)$ to all monomials $f$ to obtain the following \textit{martingale constraints} in terms of the moment sequences $[m_i]_{i \in \mathbb{N}^{n+1}}$ and $[b_i]_{i \in \mathbb{N}^{n+1}}$:

\begin{equation}
    \sum_{j \in \mathbb{N}^{n+1}}\big[ c_j(i) \cdot m_j \big] + \hat{x}_0^i - b_i = 0
    \label{eq:spacetime_mt_const}
\end{equation}

\noindent for every monomial $f(x,s) \in D(A)$, such that $f(x,s) = (x,s)^k$, $k \in \mathbb{N}^{n+1}$. 

The formulation of the moment and localizing matrix constraints remains the same as that of \cref{sec:sdp_moment_constraints}, while the maximization (minimization) variable when computing the higher order moments $\mathbb{E}[\tau^n]$ is now $n\cdot m_{0,n-1}$. The size of matrix for each degree $k$ is increased accordingly (scaled combinatorially with state dimension). Lastly, the constraint on the time dimension $t \in [0,T]$ also adds polynomials to the semi-algebraic safe and boundary sets.

\subsection{Augmentation with Redundant States}
The restriction to polynomial drift and diffusion dynamics exclude a large class of real world systems that exhibit other nonlinear behaviors in their dynamics model. In particular, physical systems operating in multidimensional space often incorporate sinusoidal dynamics which are used to specify force components acting on the system with respect to a particular coordinate frame. Sinusoids are also found in the rotation matrices used to transform agents into a global frame of reference. For example, the dynamics of the Dubins car depend on the sine/cosine of the heading angle of the car, while those of a quadcopter depend on the sine/cosine of it's roll, pitch, and yaw. 

In order to support these dynamics for $X$ (including sinusoidal and natural exponential functions), the assumption of the infinitesimal generator mapping monomial test functions $f$ to polynomials must be broken. As a result, one is unable to relax the basic adjoint equation and generate a series of constraints based on the moment sequences of $\mu_0$ and $\mu_1$. In this section, we provide a state augmentation technique to restore this desired property of the generator.

\begin{definition}
Given a stochastic process $X$ in $\mathbb{R}^n$, we say that $X$ is \textit{closed under infinitesimal generation} if $Af(x)$ is a polynomial with respect to the state variables for all monomial functions $f(x) = x^i$, $i \in \mathbb{N}^n$.
\label{def:closed_under_generation}
\end{definition}

\begin{prop}
The time dependent It\^{o} diffusion $\hat{X}$ in $\mathbb{R}^{n+1}$ described via the SDE (\ref{eq:t_aug_sys}) is closed under infinitesimal generation if for each drift term $\hat{h}_i(\hat{x})$, $0 \leq i < n+1$, and diffusion term $\hat{\sigma}_{j,k}(\hat{x})$, $0 \leq j,k < n+1$, there exists $n+1$ dimensional multi-index sets $P,Q$, such that $\hat{h}_i(\hat{x}) = \sum_{p \in P}c_p\hat{x}^p$ and $(\hat{\sigma}\hat{\sigma}^\intercal)_{j,k}(\hat{x}) = \sum_{q \in Q}c_q\hat{x}^q$, where $c_p, c_q \in \mathbb{R}$. \footnote{Assuming $dB_t$ has dimension $n+1$}
\label{prop:SDE_closed_under_generation}
\end{prop}

In the cases where the dynamics of the SDE violate the requirements in Proposition \ref{prop:SDE_closed_under_generation} we propose an augmentation technique where the state space is extended with redundant variables. The augmentation is chosen such that the expanded state space now includes the non-polynomial (w.r.t the state) terms of the drift and diffusion components, as well as possibly their derivatives. 

Let $\{c_j(X_t)\}_{j \in [0,J-1]}$ be the set of coefficients (along with possibly their derivatives) of the generator $Af(x,s)$ that violate the generator polynomial mapping assumption:

\begin{equation}
    f \mapsto Af(x,s) = \sum_{i \in \mathbb{N}^{n+1}} c_i(k) \cdot (x,s)^i\\[5pt]
    \label{eq:polynomial_assumption}
\end{equation}

Note that the $c_j$'s are partially determined by the coefficients of the partial derivatives found in the infinitesimal generator and correspond to the drift and diffusion terms of the SDE. We consider the \textit{augmented state space} $\hat{X}_t \in \mathbb{R}^{J+n+1}$: 

\begin{equation*}
\hat{X}_t = 
[X_t, t, c_0(X_t), \cdot\cdot\cdot, c_J(X_t)]^\intercal \\[7pt]
\end{equation*}

\noindent with corresponding dynamics:

\begin{equation*}
\begin{gathered}
d\hat{X}_t = \hat{h}(\hat{X}_t) dt + \hat{\sigma}(\hat{X}_t)dB_t \\[7pt]
= \begin{bmatrix}
h(X_t) \\
1\\
\partial c_0(X_t)\\
\vdots\\
\partial c_{J-1}(X_t)
\end{bmatrix} dt + \begin{bmatrix}
\sigma(X_t) \\
0\\
\sigma_0(X_t)\\
\vdots\\
\sigma_{J-1}(X_t)
\end{bmatrix} dB_t \\[7pt]
\end{gathered}
\label{eq:coeff_aug}
\end{equation*}

In the case of sinusoidal dynamics, the functions $\sin$ and $\cos$ form a length 4 cycle under the derivative operator. The cycle property allows us to consider an augmentation consisting of states which cover all unique sinusoidal frequencies and phases from the original dynamics. We can employ this characteristic to formulate a general augmentation methodology for all multidimensional SDE dynamics where the drift and diffusion terms are polynomials with respect to sinusoidal dynamics and the state variables.

\begin{theorem}
Let $X$ be a process with state $\mathbf{x} \in \mathbb{R}^n$ and sinusoidal drift and diffusion such that the dynamics are:

\begin{equation}
    dX_t = \begin{bmatrix}
    h_1(X_t) \\
    h_2(X_t) \\
    \vdots \\
    h_n(X_t)
    \end{bmatrix}dt + 
    \begin{bmatrix}
    \sigma_{1,1}(X_t)&...&\sigma_{1,d}(X_t)\\
    \sigma_{2,1}(X_t)&...&\sigma_{2,d}(X_t)\\
    \vdots\\
    \sigma_{n,1}(X_t)&...&\sigma_{n,d}(X_t)
    \end{bmatrix}dB_t
    \label{eq:theorem_v1}
\end{equation}

\noindent where

\begin{equation*}
\begin{aligned}
    h_i(X_t) &= \sum_p \alpha_p^{(i)} (\sin(\phi_p\mathbf{x}^{\gamma_p}), \cos(\psi_p\mathbf{x}^{\gamma_p}), \mathbf{x}^{\gamma_p})^{\beta_p} \\
    \sigma_{ij}(X_t) &= \sum_q \alpha_q^{(i,j)} (\sin(\phi_q\mathbf{x}^{\gamma_q}), \cos(\psi_q\mathbf{x}^{\gamma_q}), \mathbf{x}^{\gamma_q})^{\beta_q}
    \end{aligned}
\end{equation*}

\noindent where each $\alpha_{(\cdot)} \in \mathbb{R}$ is a scalar coefficient, $\beta_{(\cdot)} \in \mathbb{N}^n$ is a multi-index, $\mathbf{x}^{\gamma_{(\cdot)}}$ is a monomial with respect to the state $\mathbf{x}$ given by a multi-index $\gamma_{(\cdot)} \in \mathbb{N}^n$, and $\phi_{(\cdot)} \in \Phi$, $\psi_{(\cdot)} \in \Psi$ are a finite set of frequencies. Let $\mathbf{\hat{x}}$ denote the sinusoidal augmented state space such that:

\vspace{0.1cm}
\begin{equation*}
\begin{gathered}[]
    \mathbf{\hat{x}} = [\mathbf{x}, \sin(\phi_1\mathbf{x}^{\gamma_1}),  \sin(\phi_2\mathbf{x}^{\gamma_2}), ... \sin(\phi_m\mathbf{x}^{\gamma_m}),\\
    \sin(\psi_1\mathbf{x}^{\gamma_1}),  \sin(\psi_2\mathbf{x}^{\gamma_2}), ... \sin(\psi_m\mathbf{x}^{\gamma_m}),\\
    \cos(\phi_1\mathbf{x}^{\gamma_1}),\cos(\phi_2\mathbf{x}^{\gamma_2}),... \cos(\phi_m\mathbf{x}^{\gamma_m})\\
    \cos(\psi_1\mathbf{x}^{\gamma_1}),\cos(\psi_2\mathbf{x}^{\gamma_2}),... \cos(\psi_m\mathbf{x}^{\gamma_m})]     
\end{gathered}
\end{equation*}
\vspace{0.2cm}

\noindent where $\phi_1, \cdots \phi_m \in \Phi$ and $\psi_1, \cdots \psi_m \in \Psi$. Then, the augmented state $\mathbf{\hat{x}}$ has dimension $2(|\Phi|+|\Psi|) + n$ and the augmented system $\hat{X}$ satisfies \cref{def:closed_under_generation} --- In other words, the augmented state includes sine and cosine terms for all unique frequencies found in the dynamics of $X$.
\label{thm:SDE_aug}
\end{theorem}

\noindent \textit{Proof}: See Appendix \ref{appendix:state_aug_proof}.

\newcommand{\xhs}{x^{\gamma_z}}
\newcommand{\xh}{\mathbf{x}^{\gamma_h}}
\newcommand{\bx}{\mathbf{x}}
\renewcommand{\qedsymbol}{$\blacksquare$}

\vspace{0.3cm}

An example of obtaining the martingale constraints through state space augmentation for a time dependent SDE is given in Appendix \ref{example:state_aug}.

With the appropriate state augmentation, we may now obtain the martingale constraints (\ref{eq:spacetime_mt_const}) in terms of the moment sequences for previously unsupported nonlinear dynamics. In performing the augmentation involving sinusoidal terms, if the original system dynamics has $|\Phi| + |\Psi|$ unique modes, then the dimension of the state increases by $2(|\Phi| + |\Psi|)$. This poses a challenge for the computation of the optimization program in \ref{sec:optimization_program}. As the dynamics increase in complexity with additional sinusoidal modes, the augmentation requires more states resulting in the PSD matrix constraints growing intractable. 
\section{Reduced Constraint Optimization}
\label{sec:reduced_constraint}

In this section we present solutions to mitigate the computational challenges associated with applying the SDP moment method to complex systems.

\subsection{Reformulated Localizing Matrix Constraints}
Following the definition of the moment sequence, it is clear that the number of moments and the size of the moment/localizing matrices scale combinatorially with the dimension of the state space. This poses a challenge to the computational feasibility when high dimensional (possibly state augmented) systems are considered. In addition, the number of localizing matrix constraints grows linearly with the number of polynomials $q_i$ used to specify the safe set and its boundary (while the size of each matrix grows combinatorially). Given a complex safe set with numerous polynomials, this quickly results in an intractable number of extremely large PSD matrix constraints. Therefore, we propose replacing the sequence of localizing matrix constraints for the exit measure with a set of scalar equality constraints over the moment sequence. We add additional assumptions to the polynomials that form the semialgebaric safe set which, in practice, are easily satisfied. 

\begin{definition}
For a polynomial $q$, a \textit{critical point} $x$ is a point where the derivative of $q$ vanishes. A \textit{critical value} of $q$ is an element of the co-domain in the image of some critical point. 
\end{definition}

\begin{prop}
Suppose the semi-algebraic safe set is given by $S = \{x \mid q_i(x) \geq 0, i \in [N]\}$ and 0 is not a critical value of $q_i$ $\forall i \in [N]$, then the boundary is characterized by $\partial S = \{ x \mid \text{ there exists } i \text{ such that } q_i(x) = 0\}$.
\label{prop:safeset_boundary}
\end{prop}

\begin{proof}
Consider the case when $\forall i$ $q_i(x) \neq 0$. If $q_i(x) > 0$, then there must exist $\epsilon > 0$ such that $q_i(x+\epsilon) > 0 \in S$. As a result, $x$ cannot be on the boundary of $S$. Likewise, if $q_i(x) < 0$, then there must exist $\epsilon > 0$ such that $q_i(x+\epsilon) < 0 \notin S$, $x \notin \partial S$. Therefore, by the contrapositive, if $x \in \partial S$ then there must exist an $i$ such that $q_i(x) = 0$.

Now consider the case when there exists an $i$ such that $q_i(x) = 0$. We are given that 0 is not a critical value so $\nabla_xq_i \neq 0$, therefore there exists $\epsilon > 0$ such that either
$$
q_i(x) + \epsilon \cdot \nabla_xq_i > 0 \text{ }\text{ and }\text{ } q_i(x) - \epsilon \cdot \nabla_xq_i < 0
$$
\noindent or
$$
q_i(x) + \epsilon \cdot \nabla_xq_i < 0 \text{ }\text{ and }\text{ } q_i(x) - \epsilon \cdot \nabla_xq_i > 0
$$
As a result, $x$ must be on the boundary of $S$. 
\end{proof}

Given a safe set $S = \{ x \mid q_i(x) \geq 0 , i \in [N]\}$, let the polynomial $q'$ be the product of the $q_i$'s. Using proposition \ref{prop:safeset_boundary}, the boundary can be characterized as:

$$
\partial S = \{ x \mid q'(x) \geq 0 , -q'(x) \geq 0\}
$$

Therefore, all entries of the localizing matrix satisfy $M_k(q'b)(i,j) = 0$, yielding the set of scalar equality constraints:

$$
\sum_\alpha q'_\alpha \cdot b_{\beta(i,j) + \alpha} = 0 \quad \quad len(\textbf{M}(b)) > i,j \geq 0
$$

\noindent where $\textbf{M}(b)$ is the sequence of moments with respect to the exit measure. The new optimization program is given as follows:

\vspace{0.5cm}
\noindent \textbf{Optimization II (Reduced Constraints)} 
\begin{alignat*}{3}
& \text{Maximize}   & &(\textit{resp. Minimize}):& n\cdot m_{0,n-1}\\
\cline{1-6}
&  \text{Subject to:}&   &\hspace{0.7cm}\hat{x}_0^k + \sum_{i \in \mathbb{N}^{n+1}}c_i(k) \cdot m_i - b_k= 0&\\
&   &   &\hspace{0.7cm}M_K(m) \succcurlyeq 0& \\
&   &   &\hspace{0.7cm}M_K(b) \succcurlyeq 0&\\
&   &   &\hspace{0.7cm}M_K(q_im) \succcurlyeq 0&\\
&   &   &\hspace{0.7cm}\sum_\alpha q'_\alpha \cdot b_{\beta(i,j) + \alpha} = 0&\hspace{0.1cm}\forall k \leq K\\\\
\end{alignat*}

\noindent where the $q_i$'s are given by the polynomials of the semi-algebragic safe set and $q'$ by the product of the $q_i$'s. Due to the symmetry of the remaining moment and localizing matrices, their respective sequences of PSD constraints may be replaced with a single constraint involving the highest moment degree $K$, further reducing the memory requirements during computation. 

\subsection{Computational Consequences (Splitting Conic Solver)}
For intuition on the computational impact of the reformulated optimization, let us consider a widely used SDP solver such as the Splitting Conic Solver (SCS) --- a default solver included as part of the CVXPY convex optimization modelling language for Python \cite{ODonoghue2016ConicOV}. The algorithm consists of three main steps with the primary computational burdens falling upon: 1) Projection onto a subspace by solving a linear system with a coefficient matrix $I + Q$, and 2) projection onto a cone requiring an eigendecomposition. Suppose we are solving the moment method problem consisting of $N$ moments given by $x = [m_0, m_1, m_2...]$, up to a maximum moment degree of $K$. The semi-algebraic safe set $S$ consists of $N_q$ polynomials and the boundary $\partial S$ consists of $2N_q$ polynomials. Let the moment and localizing matrices $M_K(m)$ and $M_K(qm)$ have dimension $d_K$. 

Using Optimization I with original constraints from \cref{sec:optimization_program} (sequence of moment/localizing matrices replaced with highest degree matrix), the PSD constraints are formed using a block diagonal matrix consisting of $M_K(m)$, $M_K(b)$, $\{M_K(q_im)\ \mid i \in [N_q]\}$, and $\{M_K(q_ib)\ \mid i \in [2N_q]\}$, resulting in a PSD matrix constraint of size $(2+3N_q)d_K$ by $(2+3N_q)d_K$. On the other hand, by using Optimization II with reduced constraints, the localizing matrices of the boundary is removed, resulting in a PSD matrix constraint of size $(2+N_q)d_K$ by $(2+N_q)d_K$. This causes a reduction ($\delta$) in the size of the SCS coefficient matrix ($I+Q$) that is proportional to the square of the number of polynomials and dimension of moment/localizing matrix: $\delta = O(N_q^2d_K^2)$. 

Consider a scenario where one is computing the exit time of a multidimensional system with a number of sinusoidal modes. As state augmentation is used to incorporate each of the modes into the state space, the sequence of moments increases combinatorially in length with respect to the new state dimension. This in turn directly increases the dimensions of the moment and localizing matrices ($d_K$). The combinatorial growth in $d_K$ is then squared to contribute an even larger impact on the coefficient matrix of SCS. The growth in computational costs are further exacerbated when complex safe sets consisting of numerous polynomials are used. As a result, even minor applications of state augmentation can have a dramatic impact on the difference in computational feasibility of the moment method when comparing between the original and reduced constraint formulations. 
\section{Examples}
\label{sec:experiments}

\renewcommand{\arraystretch}{1.2}
\begin{table*}
        \caption{BM Exit Time Moments ($K = 8$)}
        \label{table:brownian_motion_moments_k8}
        \centering
        \begin{tabular}{c|c|ll|ll}
            \toprule
            & & \hspace{1.8cm} Original & \hspace{-0.95cm}Constraints & \hspace{1.8cm} Reduced & \hspace{-0.95cm}Constraints \\\cline{3-6}
            Moment & Analytical Value & \hspace{0.5cm} Lower Bound \hspace{1.2cm} & Upper Bound \hspace{0.5cm} & \hspace{0.5cm} Lower Bound \hspace{1.2cm}  & Upper Bound \hspace{0.5cm} \\
            \midrule
            1 & 0.25000 & \hspace{0.6cm}0.24999 & 0.25003 & \hspace{0.6cm}0.25000 & 0.25000 \\ 
            2 & 0.10417 & \hspace{0.6cm}0.10410 & 0.10421 & \hspace{0.6cm}0.10416 & 0.10418 \\
            3 & 0.06354 & \hspace{0.6cm}0.06339 & 0.06389 & \hspace{0.6cm}0.06348 & 0.06434 \\
            4 & 0.05153 & \hspace{0.6cm}0.04487 & 0.06690 & \hspace{0.6cm}0.05131 & 0.05258 \\
            5 & 0.05221 & \hspace{0.6cm}0.03460 & 0.30626 & \hspace{0.6cm}0.05133 & 0.06491 \\
            6 & 0.06348 & \hspace{0.6cm}0.02910 & $-$ & \hspace{0.6cm}0.05861 & 0.20670 \\
            \bottomrule
        \end{tabular}
        \caption*{A comparison of the upper and lower bounds of the first six moments of a time-space Brownian motion. The bounds are computed using both original and reduced constraints (Optimization I and Optimization II, respectively). A moment sequence with maximum degree $K=8$ is used in both scenarios. Dashes ($-$) indicate settings where SCS did not converge. The SDP with reduced constraints demonstrates tighter bounds, particularly for higher order moments.}
\end{table*}

In this section we provide numerical examples of computing the exit time moments of systems of varying complexity. We highlight the usage of state augmentation to support non-polynomial physical systems and the additional hurdles in computation it brings along.

\subsection{Time-Space Brownian Motion}

We first demonstrate the scalability of the reduced constraints SDP to higher order exit time moments and longer moment sequences through a two-dimensional time-space Brownian motion example. Let $Y_t = y_0 + W_t$ where $W_t$ is a one-dimensional Brownian motion. The time-space process is given by $X = \{(t, Y_t)\}$. The generator is given by:

$$
Af(t,y) = \frac{\partial f}{\partial t}(t,y) + \frac{1}{2}\frac{\partial ^2 f}{\partial y^2}(t,y)
$$

\noindent As the generator is a polynomial with respect to the state, additional state augmentation is not required. 

\renewcommand{\arraystretch}{1.2}
\begin{table}[!b]
        \caption{BM Exit Time Moments ($K=14$)}
        \label{table:brownian_motion_moments_k14}
        \centering
        \begin{tabular}{c|c|ll}
            \toprule
            & & \hspace{0.9cm} Reduced & \hspace{-0.40cm}Constraints \\\cline{3-4}
            \text{} Moment & Analytical Value & \hspace{0.2cm}Lower Bound \hspace{0.2cm} & Upper Bound \hspace{0.2cm}\\
            \midrule
            1 & 0.25000 & \hspace{0.2cm}0.24999 & 0.25003 \\ 
            2 & 0.10417 & \hspace{0.2cm}0.10415 & 0.10418 \\ 
            3 & 0.06354 & \hspace{0.2cm}0.06344 & 0.06357 \\
            4 & 0.05153 & \hspace{0.2cm}0.05135 & 0.05155 \\
            5 & 0.05221 & \hspace{0.2cm}0.05185 & 0.05233 \\
            6 & 0.06348 & \hspace{0.2cm}0.06248 & 0.06387 \\
            \bottomrule
        \end{tabular}
        \caption*{The upper and lower bounds for a time-space Brownian motion are computed using the SDP method with reduced constraints (Optimization II) and moment sequence with maximum degree $K=14$. With a larger moment sequence, we see tighter bounds for the higher order moments. SCS is unable to converge for all six moments when using an SDP with original constraints (Optimization I).}
        \vspace{0.1cm}
\end{table}

We consider the exit of the process from a safe set given by $S = \{(t,y) \mid T \geq t \geq 0, 1 \geq y \geq 0 \}$. The safe set contains the space interval [0,1] and a finite time interval up to time $T$. The initial condition is given by $y_0 = 0.5$. \cref{table:brownian_motion_moments_k8,,table:brownian_motion_moments_k14} show the computed upper and lower bounds of the first six moments of the exit time. 

\vspace{0.3cm}

Using Optimization I with original constraints, we see that the solver returns bounds with minimal spread for lower order moments of the exit time (\cref{table:brownian_motion_moments_k8}). Above the third moment however, the spread between the lower and upper bounds begin to increase due to numerical instabilities and an insufficient moment sequence length. With a maximum moment degree $K=8$, the SDP with original constraints is unable to produce an upper bound for the sixth moment when using SCS. On the other hand, when using Optimization II with reduced constraints, we are able to compute values for all six moments and see a smaller spread between upper and lower bounds. 

A larger moment sequence is required to compute accurate bounds for higher order moments. In this example, we consider a maximum moment degree up to $K=14$. \cref{table:brownian_motion_moments_k14} shows the trade-off between accuracy and computational feasibility that needs to be made when using the SDP with original constraints. Under the scenario with a larger moment sequence, the original formulation fails to converge for all six moments when using SCS. In comparison, the reduced constraint SDP continues to provide bounds for all moments. As expected, we are able to obtain tighter bounds versus those computed with $K=8$ (\cref{table:brownian_motion_moments_k8}). 

\subsection{Spring Mass Damper with Variable Damping Rate}

We consider a spring mass damper system where the mass sits vertically above the spring and damper with the following parameters:
\begin{itemize}
    \item Spring constant $k_s=5.0$ 
    \item Object mass $m_s=1.0$
    \item Static damper constant $k_c=1.0$
\end{itemize}

To demonstrate the redundant state augmentation technique, we consider a variable damper force subject to noise and proportional to both the static damper constant and a sinusoidal term with respect to the position of the mass. A diagram of the setup is shown in Fig. \ref{fig:smd_system_diagram}. The state space $X_t$ is defined as follows:
\begin{equation*}
    X_t = [x, v, t]^\intercal \\[5pt]
\end{equation*}

\renewcommand{\arraystretch}{1.2}
\begin{table*}[!t]
        \caption{Exit Time Bounds of State Augmented Spring/Damper System}
        \caption*{Safe Set $S_1$: $x \in [-2, 0]$}
        \label{table:spring_damper_sdp_s1}
        \centering
        \begin{tabular}{l|llll|llll}
            \toprule
            & & \hspace{0.425cm} Original & \hspace{-0.425cm}Constraints & & & \hspace{0.425cm} Reduced & \hspace{-0.425cm} Constraints & \\\cline{2-9}
            Max Degree & Lower Bound & Upper Bound & LB Runtime & UB Runtime & Lower Bound & Upper Bound & LB Runtime & UB Runtime\\
            \midrule
            $K=4$ & 0.10522 & 49.97546 & 2.84 & 7.03 & 0.02006 & $-$ & 1.28 & $-$\\ 
            $K=6$ & $-$ & $-$ & $-$ & $-$ & 0.10669 & 4.86442 & 14.53 & 572.98 \\ 
            $K=8$ & $-$ & $-$ & $-$ & $-$ & 0.20143 & 2.96940 & 246.18 & 402.81 \\
            $K=10$ & $-$ & $-$ & $-$ & $-$ & 0.65176 & 1.15057 & 1467.34 & 1570.41 \\
            \midrule
            Simulation: \quad & 1.00633 & & \\
            \bottomrule
        \end{tabular}
\end{table*}

\renewcommand{\arraystretch}{1.2}
\begin{table*}[!t]
        \caption*{Safe Set $S_2$: $x \in [-2.5, 0]$}
        \label{table:spring_damper_sdp_s2}
        \centering
        \begin{tabular}{l|llll|llll}
            \toprule
            & & \hspace{0.425cm} Original & \hspace{-0.425cm}Constraints & & & \hspace{0.425cm} Reduced & \hspace{-0.425cm} Constraints & \\\cline{2-9}
            Max Degree & Lower Bound & Upper Bound & LB Runtime & UB Runtime & Lower Bound & Upper Bound & LB Runtime & UB Runtime\\
            \midrule
            $K=4$ & 0.47416 & $-$ & 37.75 & $-$ & 0.11550 & $-$ & 28.74 & $-$\\ 
            $K=6$ & $-$ & $-$ & $-$ & $-$ & 2.14749 & 40.74567 & 200.51 & 1026.80 \\ 
            $K=8$ & $-$ & $-$ & $-$ & $-$ & 3.20881 & 10.95551 & 2415.07 & 2259.54 \\
            $K=10$ & $-$ & $-$ & $-$ & $-$ & 4.60201 & 9.87456 & 6143.05 & 36855.54 \\
            \midrule
            Simulation: \quad & 9.85540 & & \\
            \bottomrule
        \end{tabular}
        \caption*{A comparison of the upper/lower bounds of the exit time for a state augmented system computed using the SDP method with original and reduced constraints. All values have units in seconds. Dashes ($-$) indicate settings where SCS fails to converge. For both safe sets, the SDP with reduced constraints continued to yield results when a larger moment sequence is considered while the SDP with original constraints failed to converge.}
\end{table*}

Here, $x$ is the vertical position of the mass, $v$ is its velocity, and $t$ is the time. The system dynamics are given by: 

\begin{equation*}
    \begin{gathered}
    dX_t = 
    \begin{bmatrix}
    v\\
    -\frac{k_s}{m_s}x-g+\frac{k_c}{m_s}v\sin(x)\\
    1
    \end{bmatrix}dt + \begin{bmatrix}
    0\\
    \frac{k_c}{m_s}\\
    0
    \end{bmatrix}dB_t \\[7pt]
    \end{gathered}
\end{equation*}

In order to produce a generator that maps monomial test functions $f$ to polynomials with respect to the state variables, we consider the following state augmentation:
\begin{equation*}
    \begin{gathered}
    \hat{X}_t = [x,v,t,\sin(x),\cos(x)]^\intercal\\[10pt]
    d\hat{X}_t = 
    \begin{bmatrix}
    v\\
    -\frac{k_s}{m_s}x-g+\frac{k_c}{m_s}v\sin(x)\\
    1 \\
    v\cos(x)\\
    v\sin(x)
    \end{bmatrix}dt + \begin{bmatrix}
    0\\
    \frac{k_c}{m_s}\\
    0\\0\\0
    \end{bmatrix}dB_t\\[7pt]
    \end{gathered}
\end{equation*}

We consider two safe sets with differing exit times: $S_1$ where all safe states $x \in [-2,0]$, and $S_2$ where $x \in [-2.5,0]$. As before, the exit time is made finite with a time horizon $T$. The initial values for vertical position and velocity are $-\frac{9.81}{k_s}$ and $0$, respectively. The upper and lower bounds of the exit time of the augmented SDE is calculated through the SDP formulation using both the original and reduced constraint sets. The Splitting Conic Solver (SCS) with CVXPY is used for all calculations.

\begin{figure}[!h]
    \centering
    \includegraphics[width=0.6\columnwidth]{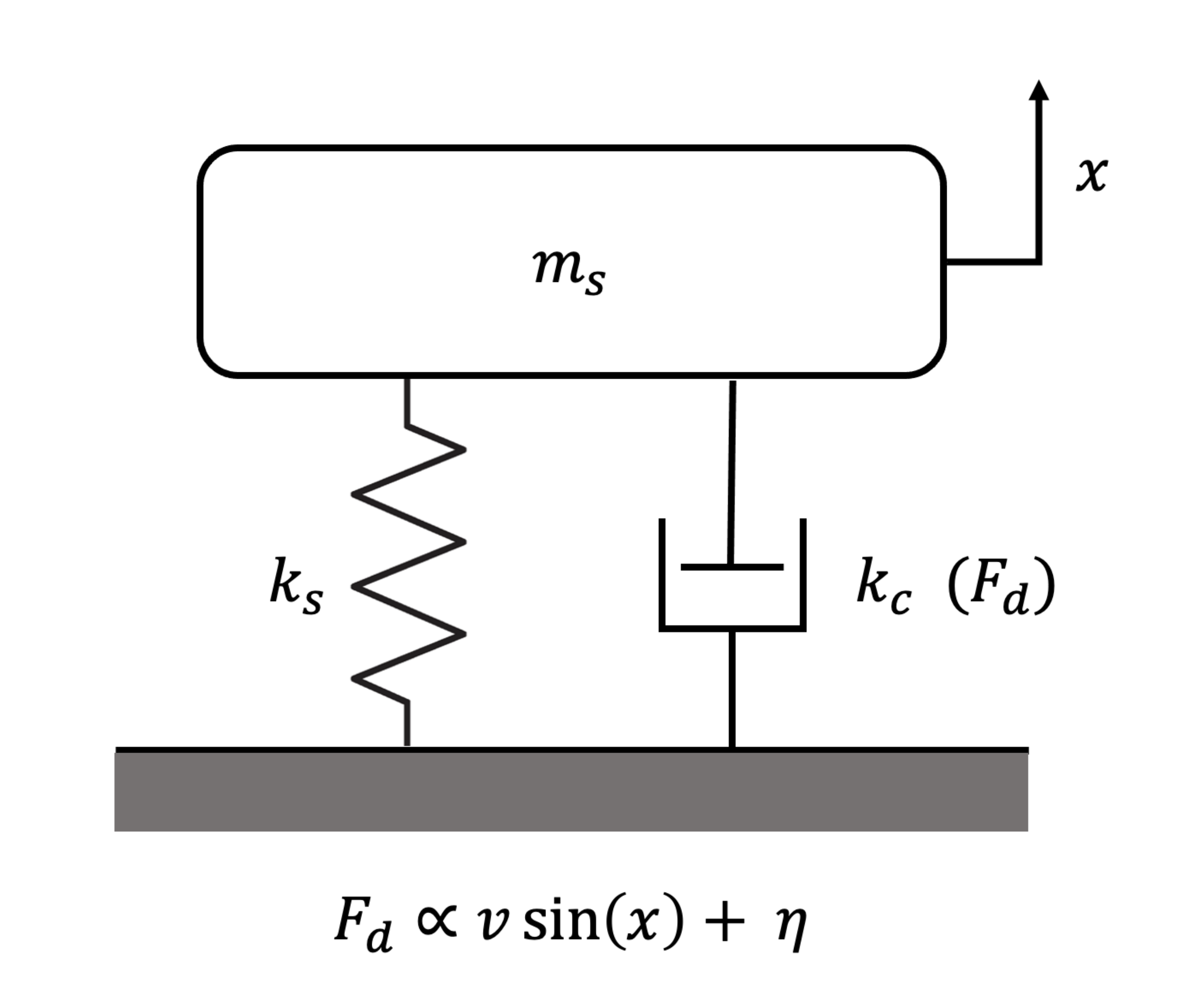}
    \caption{Spring mass damper system with variable damping rate. Damper force $F_d$ is proportional to a sinusoidal function of mass position.}
    \label{fig:smd_system_diagram}
\end{figure}

\cref{table:spring_damper_sdp_s1} shows the bounds on the expected exit time of the system calculated using the SDP method. The max degree $K$ represents the highest order moment considered in the optimization program. All other values in the tables have units of seconds. The observations from the tables provide evidence that the reduced constraints of Optimization II benefits the application of the moment method on higher dimensional systems that require state augmentation. With the original SDP constraints, the solver can only handle the state augmented system up to a maximum moment degree of $K=4$, resulting in a very loose bound on the exit time which inhibits its use for evaluating whether the system's true expected exit time satisfies some desired safety requirements. For degrees above $K=4$, SCS does not converge in the provided number of iterations. The solver continues to struggle when the expected exit time is large (\cref{table:spring_damper_sdp_s1} Safe Set $S_2$) which results in a greater discrepancy in magnitude between moments and higher numerical instability.

By replacing the large PSD matrix constraints associated with the localizing matrices, the solver can compute sequences with significantly higher moment degrees (up to $K=10$) and provide more accurate bounds on the exit time. It should be noted that using the reduced constraint SDP requires the the safe set boundary $\partial S$ to be a product of polynomials, resulting in constraints that contain higher order moments. As such, $K=4$ represents too small of a moment sequence to accurately capture the constraints. The per iteration runtime of SCS shows a non-linear growth trend for each increased moment degree. In some cases, we saw SCS demonstrate a greater total runtime for lower moment degrees due to the solver taking a drastically different number of iterations to converge (e.g. the computation of the upper bound of $K=6$ versus $K=8$ for $S_1$). 

When comparing the original and reduced constraint SDP formulations, the latter provides significantly more useful information regarding the safety behavior of the state augmented system. The columns on the right-hand side of \cref{table:spring_damper_sdp_s1} show much tighter bounds for both safe sets $S_1$ and $S_2$. As was seen previously when using the original constraints, the larger expected exit time of $S_2$ poses SCS with a more difficult problem to solve, however unlike before, we continue to see convergence of the solver for both upper and lower bounds when using the reduced constraints. 

\section{Conclusion}
\label{sec:conclusion}

In this paper we considered a safety analysis of stochastic systems through a moment based exit time method. Our formulation first considers the martingale problem and uses it to define a linear evolution equation linking the occupation and exit measures of the stochastic process under study. When considering processes with appropriate (polynomial) generators, this evolution can be relaxed to a series of conditions involving the moments of the occupation and exit measures. Together with appropriate SDP moment conditions, a convex optimization problem is formed to compute bounds on the exit time moments of the process. Noting the strong assumptions on the system dynamics required for the moment method, we propose a state space augmentation technique to support a broader class of systems. The use of state augmentation expands the moment method to characterize the safety behavior of a wide range of physical systems beyond the polynomial dynamics required by the original approach. We discuss the trade off in computational feasibility that comes with state space augmentation and propose a reformulation of the optimization constraints. The method grants an easily automated procedure for simplifying the large PSD matrix constraints associated with complex dynamics and safe sets --- greatly improving scalability into higher dimensional systems. Taken together, we showed that our methods can be applied to a broader class of polynomial and non-polynomial dynamics, and presented scenarios where a consideration of exit time moments grants useful insight into the safety of the system. 
\newpage




\bibliographystyle{IEEEtran}
\bibliography{references}

\newpage 
\appendix

\subsection{State Augmentation}
\label{appendix:state_aug_proof}
We present here the proof of \cref{thm:SDE_aug} for the general $n$-dimensional SDE.

\begin{proof} We denote each additional state in the augmented state space as $f(\xi_h\mathbf{x}^{\gamma_h})$. Recall that $\mathbf{x}^{\gamma_{(\cdot)}}$ is a monomial with respect to the state. For an arbitrary $h, 0 \leq h < m$, the dynamics of the state $f(\xi_h\mathbf{x}^{\gamma_h})$ is given by: 
\begin{equation*}
    \begin{aligned}
    df(\xi_h\xh) = f'(\xi_h\xh) \cdot \xi_h \cdot d\xh \\ 
    + \frac{\xi_h^2}{2} \sum_{i,j} \frac{\partial^2 f}{\partial x_i \partial x_j} \cdot (dx_idx_j)
    \end{aligned}
    \label{proof_eq:1}
\end{equation*}
\noindent where 
\begin{equation}
    \begin{aligned}
    d\xh &= \sum_{i=1}^n \big( \mathbbm{1}_{x_i \in \xh} \cdot \prod_{j \neq i} x_j^{\gamma_h[j]} \cdot d x_i^{\gamma_h[i]} \big) \\ \nonumber
    &= \sum_{i=1}^n \big( \mathbbm{1}_{x_i \in \xh} \cdot \prod_{j \neq i} x_j^{\gamma_h[j]} \cdot \gamma_h[i]x_i^{\gamma_h[i]-1} \cdot d x_i \big) \\ \nonumber
    \end{aligned}
\end{equation}

\vspace{-0.35cm}

\noindent We note that by construction, $dx_i$ has drift ($h_i$) and diffusion ($\sigma_{i,(\cdot)}$) that is polynomial w.r.t. the sinusoidal terms $\sin(\phi_{(\cdot)}\bx^{\gamma_{(\cdot)}})$, $\cos(\phi_{(\cdot)}\bx^{\gamma_{(\cdot)}})$, $\sin(\psi_{(\cdot)}\bx^{\gamma_{(\cdot)}})$, $\cos(\psi_{(\cdot)}\bx^{\gamma_{(\cdot)}})$, and state $\bx$. Furthermore, the augmented state space $\mathbf{\hat{x}}$ includes all sinusoidal terms $\sin(\phi_{(\cdot)}\bx^{\gamma_{(\cdot)}})$, $\cos(\phi_{(\cdot)}\bx^{\gamma_{(\cdot)}})$, $\sin(\psi_{(\cdot)}\bx^{\gamma_{(\cdot)}})$, $\cos(\psi_{(\cdot)}\bx^{\gamma_{(\cdot)}})$, up to monomials of degree $\gamma_m$ where $\gamma_m$ is greater than the order of the highest monomial $\bx^{\gamma_{(\cdot)}}$ in the original dynamics of $X$. Therefore $dx_i$ is polynomial w.r.t. the augmented state space $\mathbf{\hat{x}}$. We denote the drift and diffusion of $dx_i$ as $p_i^{(1)}$ and $p_i^{(2)}$, respectively:

\begin{equation*}
    \begin{gathered}
    d\xh = \sum_{i=1}^n \big( \mathbbm{1}_{x_i \in \xh} \cdot \gamma_h[i]x_i^{\gamma_h[i]-1} \cdot \\ \prod_{j \neq i} x_j^{\gamma_h[j]} \cdot 
    [p_i^{(1)}dt + \langle p_i^{(2)},dB_t \rangle\ ] \big) \\
    \end{gathered}
\end{equation*}
\begin{equation*}
    \begin{gathered}
    = \sum_{i=1}^n \big( \mathbbm{1}_{x_i \in \xh} \cdot \gamma_h[i]x_i^{\gamma_h[i]-1} \cdot \prod_{j \neq i} x_j^{\gamma_h[j]} \cdot p_i^{(1)} \big) dt \\
    + \sum_{i=1}^n \big( \mathbbm{1}_{x_i \in \xh} \cdot \gamma_h[i]x_i^{\gamma_h[i]-1} \cdot \prod_{j \neq i} x_j^{\gamma_h[j]} \cdot \langle p_i^{(2)},dB_t \rangle\big)
    \end{gathered}
\end{equation*}

\vspace{0.3cm}

\noindent We define the following \textit{drift sub-matrix} quantities:
\begin{equation*}
\begin{gathered}
    \bm{h} = [h_1, h_2, \dots h_n]^\intercal
\end{gathered}
\end{equation*}

\begin{equation*}
    \begin{gathered}[]
    \bm{h_{\sin}} = [h_{\sin(\phi_1\bx^{\gamma_1})}, \dots h_{\sin(\phi_m\bx^{\gamma_m})}, \\
    h_{\sin(\psi_1\bx^{\gamma_1})}, 
    \dots h_{\sin(\psi_m\bx^{\gamma_m})}]^\intercal
    \end{gathered}
\end{equation*}

\vspace{0.3cm}

\begin{equation*}
    \begin{gathered}[]
    \bm{h_{\cos}} = [h_{\cos(\phi_1\bx^{\gamma_1})}, \dots h_{\cos(\phi_m\bx^{\gamma_m})}, \\
    h_{\cos(\psi_1\bx^{\gamma_1})}, 
    \dots h_{\cos(\psi_m\bx^{\gamma_m})}]^\intercal
    \end{gathered}
\end{equation*}

\vspace{0.3cm}

\noindent Next we define the following \textit{diffusion sub-matrix} quantities:

\begin{equation*}
    \begin{gathered}
    \bm{\sigma} = 
    \begin{bmatrix}
    \sigma_{1,1}&...&\sigma_{1,d}\\
    \vdots&\ddots&\vdots\\
    \sigma_{n,1}&...&\sigma_{n,d}\\
    \end{bmatrix}
    \end{gathered}\\
\end{equation*}

\begin{equation*}
    \begin{gathered}
    \bm{\sigma_{\sin}} = 
    \begin{bmatrix}
    \sigma_{\sin(\phi_1\bx^{\gamma_1}),1}&...&\sigma_{\sin(\phi_1\bx^{\gamma_1}),d}\\ 
    \vdots&\ddots&\vdots\\
    \sigma_{\sin(\phi_m\bx^{\gamma_m}),1}&...&\sigma_{\sin(\phi_m\bx^{\gamma_m}),d}\\
    \sigma_{\sin(\psi_1\bx^{\gamma_1}),1}&...&\sigma_{\sin(\psi_1\bx^{\gamma_1}),d}\\ 
    \vdots&\ddots&\vdots\\
    \sigma_{\sin(\psi_m\bx^{\gamma_m}),1}&...&\sigma_{\sin(\psi_m\bx^{\gamma_m}),d}\\
    \end{bmatrix}
    \end{gathered}\\
\end{equation*}

\begin{equation*}
    \begin{gathered}
    \bm{\sigma_{\cos}} = 
    \begin{bmatrix}
    \sigma_{\cos(\phi_1\bx^{\gamma_1}),1}&...&\sigma_{\cos(\phi_1\bx^{\gamma_1}),d}\\ 
    \vdots&\ddots&\vdots\\
    \sigma_{\cos(\phi_m\bx^{\gamma_m}),1}&...&\sigma_{\cos(\phi_m\bx^{\gamma_m}),d}\\
    \sigma_{\cos(\psi_1\bx^{\gamma_1}),1}&...&\sigma_{\cos(\psi_1\bx^{\gamma_1}),d}\\ 
    \vdots&\ddots&\vdots\\
    \sigma_{\cos(\psi_m\bx^{\gamma_m}),1}&...&\sigma_{\cos(\psi_m\bx^{\gamma_m}),d}
    \end{bmatrix}
    \end{gathered}\\
\end{equation*}

\noindent where 
    
\begin{equation}
    \begin{gathered}
    h_{\sin(\xi_{(\cdot)}\bx^{\gamma_{(\cdot)}})} = \sin(\xi_{(\cdot)}\bx^{\gamma_{(\cdot)}})' \cdot \xi_{(\cdot)} \cdot \\ \sum_{i=1}^n \Big[ \mathbbm{1}_{x_i \in \bx^{\gamma_{(\cdot)}}} \cdot \gamma_{(\cdot)}[i]x_i^{\gamma_{(\cdot)}[i]-1} \cdot \prod_{j \neq i} x_j^{\gamma_{(\cdot)}[j]} \cdot h_i \Big] \\
    + \frac{\xi_{(\cdot)}^2}{2} \sum_{i,j \in n} \Big[ \frac{\partial^2 \sin(\xi_{(\cdot)}\bx^{\gamma_{(\cdot)}})}{\partial x_i \partial x_j} \cdot \sum_{l=1}^d \sigma_{i,l} \cdot \sigma_{j,l} \Big]
    \end{gathered}
    \label{eq_proof:h_sin}
\end{equation}

\vspace{0.3cm}

\begin{equation}
    \begin{gathered}
    \sigma_{\sin(\xi_{(\cdot)}\bx^{\gamma_{(\cdot)}}),k} = \sin(\xi_{(\cdot)}\bx^{\gamma_{(\cdot)}})' \cdot \xi_{(\cdot)}\\ \sum_{i=1}^n \Big[ \mathbbm{1}_{x_i \in \bx^{\gamma_{(\cdot)}}} \cdot  \gamma_{(\cdot)}[i]x_i^{\gamma_{(\cdot)}[i]-1} \cdot \prod_{j \neq i} x_j^{\gamma_{(\cdot)}[j]} \cdot \sigma_{i,k} \Big]
    \end{gathered}
    \label{eq_proof:sigma_sin}
\end{equation}

\vspace{0.3cm}

\begin{equation}
    \begin{gathered}
    h_{\cos(\xi_{(\cdot)}\bx^{\gamma_{(\cdot)}})} = \cos(\xi_{(\cdot)}\bx^{\gamma_{(\cdot)}})' \cdot \xi_{(\cdot)} \\ \sum_{i=1}^n \Big[ \mathbbm{1}_{x_i \in  \bx^{\gamma_{(\cdot)}}} \cdot \gamma_{(\cdot)}[i]x_i^{\gamma_{(\cdot)}[i]-1} \cdot \prod_{j \neq i} x_j^{\gamma_{(\cdot)}[j]} \cdot h_i \Big] \\
    + \frac{\xi_{(\cdot)}^2}{2} \sum_{i,j \in n} \Big[ \frac{\partial^2 \cos(\xi_{(\cdot)}\bx^{\gamma_{(\cdot)}})}{\partial x_i \partial x_j} \cdot \sum_{l=1}^d \sigma_{i,l} \cdot \sigma_{j,l} \Big]
    \end{gathered}
    \label{eq_proof:h_cos}
\end{equation}

\vspace{0.3cm}

\begin{equation}
    \begin{gathered}
    \sigma_{\cos(\xi_{(\cdot)}\bx^{\gamma_{(\cdot)}}),k} = \cos(\xi_{(\cdot)}\bx^{\gamma_{(\cdot)}})' \cdot \xi_{(\cdot)} \\ \sum_{i=1}^n \Big[ \mathbbm{1}_{x_i \in \bx^{\gamma_{(\cdot)}}} \cdot  \gamma_{(\cdot)}[i]x_i^{\gamma_{(\cdot)}[i]-1} \cdot \prod_{j \neq i} x_j^{\gamma_{(\cdot)}[j]} \cdot \sigma_{i,k} \Big]
    \end{gathered}
    \label{eq_proof:sigma_cos}
\end{equation}

\vspace{0.5cm}

\noindent The dynamics of the augmented SDE is now given by:

\begin{equation*}
\begin{aligned}
d\hat{X_t} =\begin{bmatrix}
\bm{h} \\
\bm{h_{\sin}}\\
\bm{h_{\cos}}\\
\end{bmatrix} dt + \begin{bmatrix}
\bm{\sigma} \\
\bm{\sigma_{\sin}} \\
\bm{\sigma_{\cos}} 
\end{bmatrix} dB_t \\[7pt]
\end{aligned}
\label{example:state_space_aug}
\end{equation*}

\noindent We note that all terms in Eqs. (\ref{eq_proof:h_sin})-(\ref{eq_proof:sigma_cos}) are polynomial w.r.t. the augmented state space $\mathbf{\hat{x}}$, in addition, as polynomials are closed under addition and multiplication, the resulting drift and diffusion terms corresponding to the augmented state is also polynomial w.r.t. $\mathbf{\hat{x}}$. The remaining dynamics $h_1,...,h_n$ and $\sigma_{1,1},...,\sigma_{n,d}$ have already been shown to be polynomial w.r.t. $\mathbf{\hat{x}}$. We now apply \cref{eq:gen_th_ito} to obtain the generator $Af(\hat{x})$ of the augmented system. The monomials of $\mathbf{\hat{x}}$ are closed under differentiation w.r.t. $x \in \mathbf{\hat{x}}$ which gives us that $\frac{\partial f}{\partial x_i}$ and $\frac{\partial ^2 f}{\partial x_i \partial x_j}$ are monomials of $\mathbf{\hat{x}}$ for all $x_i,x_j \in \mathbf{\hat{x}}$, and test functions $f = (\mathbf{\hat{x}})^\beta, \beta \in \mathbb{N}^{|\mathbf{\hat{x}}|}$. As a result, we see that \cref{eq:gen_th_ito} applied to the augmented system $\hat{X}$ for monomial test functions yields a generator consisting of the sum of products between polynomials and monomials ($h_{(\cdot)}, \sigma_{(\cdot)}\sigma_{(\cdot)}^{\intercal}, \frac{\partial f}{\partial x_i}, \frac{\partial ^2 f}{\partial x_i \partial x_j}$) w.r.t. the augmented state space. Again, following the closure properties of polynomials, the resulting generator is polyonmial w.r.t. the augmented state space.
\end{proof}

\subsection{Examples}
\label{example:state_aug}
\theoremstyle{definition}
\newtheorem{exmp}{Example}[]
\begin{exmp}
We obtain the martingale constraints through state space augmentation for the time dependent SDE:

\vspace{0.2cm}

\begin{equation*}
\begin{aligned}
dX_t &= \begin{bmatrix}
dx\\
dt
\end{bmatrix} =\begin{bmatrix}
\sin(x) \\
1\\
\end{bmatrix} dt + \begin{bmatrix}
\cos(x) \\
0
\end{bmatrix} dB_t \\[7pt]
\end{aligned}
\label{example:state_space_aug}
\end{equation*}

\vspace{0.2cm}

\noindent Following (\ref{eq:gen_th_ito}), the generator of the system $Af$ is given by:

\vspace{0.2cm}

\begin{equation*}
    Af = \sin(x)\frac{\partial f}{\partial x} + \frac{\partial f}{\partial t} + \frac{1}{2}\cos^2(x)\frac{\partial^2f}{ \partial x^2}\\[5pt]
\end{equation*}

\vspace{0.2cm}

\noindent We see that due to the coefficients $\sin(x)$ and $\cos^2(x)$, we are unable to express the generator $Af$ as a polynomial with respect to the state space $X_t = [x,t]^\intercal$. Thus we add redundant states to augment the original state space. The augmented state space $\hat{X_t}$ is given by:

\begin{equation*}
    \hat{X}_t = [x,t,\sin(x),\cos(x)]^\intercal\\[6pt]
\end{equation*}

\noindent The dynamics of the augmented SDE are:

\vspace{-0.1cm}
\begin{equation*}
    \begin{gathered}
    d\hat{X}_t = \hat{h}(\hat{X}_t)dt + \hat{\sigma}(\hat{X}_t)dB_t \\[8pt]
    = \begin{bmatrix}
    \sin(x)\\1\\
    \cos(x)\sin(x) \big[ 1- \frac{1}{2}\cos(x) \big] \\
    -\sin^2(x) - \frac{1}{2}\cos^3(x)
    \end{bmatrix}dt + \begin{bmatrix}
    \cos(x)\\0\\\cos^2(x)\\-\sin(x)\cos(x)
    \end{bmatrix}dB_t\\[8pt]
    \end{gathered}
\end{equation*}

\noindent The SDE now satisfies the conditions in Proposition \ref{prop:SDE_closed_under_generation}. The new generator $\hat{A}f$ can then be obtained through (\ref{eq:gen_th_ito}). For monomial test functions $f$, the generator is a polynomial with respect to the augmented state $\hat{X}_t$. We can apply (\ref{eq:spacetime_mt_const}) to derive the corresponding martingale constraints. 

\end{exmp}

\end{document}